\newif\ifarxiv\arxivtrue
\newif\ifllncs\llncsfalse

\ifarxiv

\documentclass[11pt]{article}
\usepackage[left=2.5cm,top=2cm,right=2.5cm,bottom=2cm]{geometry}
\usepackage[colorlinks,linkcolor=blue,citecolor=blue,urlcolor=blue]{hyperref}
\usepackage{etex}
\usepackage{amssymb}
\usepackage{amsmath}
\usepackage{amsfonts,mathabx}
\usepackage{amsthm}
\usepackage{comment}
\usepackage{enumerate}
\usepackage[linesnumbered, ruled, lined, boxed,]{algorithm2e}
\usepackage{authblk}
\usepackage[T1]{fontenc}
\usepackage{listings}
\usepackage{booktabs}
\usepackage[toc,page]{appendix}
\usepackage{setspace}
\usepackage[textsize=tiny,textwidth=2.2cm,shadow]{todonotes}

\usepackage{arydshln}
\usepackage{mathtools}

\newtheorem{theorem}{Theorem}[section]
\newtheorem{definition}[theorem]{Definition}
\newtheorem{lemma}[theorem]{Lemma}
\newtheorem{corollary}[theorem]{Corollary}

\newtheorem{example}[theorem]{Example}

\RequirePackage{fix-cm}
\usepackage{graphicx}
\usepackage{amssymb}
\usepackage{enumerate}
\usepackage{etex,url}
\usepackage{authblk}
\usepackage[T1]{fontenc}
\usepackage{listings}
\usepackage{booktabs}
\usepackage[toc,page]{appendix}
\usepackage{arydshln}
\usepackage{mathtools}

\newcommand{\PH}{\textsc{PacketHalver}}

\newcommand{\ADD}{\textsc{Add}}
\newcommand{\RESTORE}{\textsc{Restore}}
\newcommand{\N}{\mathbb{N}}
\newcommand{\Z}{\mathbb{Z}}
\newcommand{\M}{\mathcal{M}}
\renewcommand{\S}{\mathcal{S}}
\newcommand{\R}{\mathbb{R}}
\newcommand{\Q}{\mathbb{Q}}
\newcommand{\G}{\mathcal{G}}
\renewcommand{\H}{\mathcal{H}}

\usepackage{xcolor}
\fi

\ifllncs

\RequirePackage{fix-cm}
\documentclass[smallextended]{svjour3}       
\smartqed  
\usepackage{graphicx}
\usepackage{amssymb}
\usepackage[figure,linesnumbered,vlined]{algorithm2e}
\usepackage{todonotes}

\usepackage{comment}
\usepackage{enumerate}
\usepackage{etex,url}
\usepackage{authblk}
\usepackage[T1]{fontenc}
\usepackage{listings}
\usepackage{booktabs}
\usepackage[toc,page]{appendix}
\usepackage{arydshln}
\usepackage{mathtools}
\usepackage[left=3.5cm,top=3cm,right=3.95cm,bottom=2.5cm]{geometry}

\newcommand{\PH}{\textsc{PacketHalver}}

\newcommand{\ADD}{\textsc{Add}}
\newcommand{\RESTORE}{\textsc{Restore}}
\newcommand{\N}{\mathbb{N}}
\newcommand{\Z}{\mathbb{Z}}
\newcommand{\M}{\mathcal{M}}
\renewcommand{\S}{\mathcal{S}}
\newcommand{\R}{\mathbb{R}}
\newcommand{\Q}{\mathbb{Q}}
\newcommand{\G}{\mathcal{G}}
\renewcommand{\H}{\mathcal{H}}

\usepackage{xcolor}
\usepackage{pgfplots}
\pgfplotsset{every axis/.append style={
                    label style={font=\tiny},
                    tick label style={font=\tiny}  
                    }}

\pgfplotsset{
  compat=newest,
  xlabel near ticks,
  ylabel near ticks
}

\author{Tobias Harks     \and Veerle Tan-Timmermans    
}

\institute{  T. Harks \at
            Department of Mathematics, Augsburg University
           \and V. Timmermans \at
              Department of Management Science, RWTH Aachen \\
              Tel.: +316-53958673\\
              \email{veerletimmermans@gmail.com}   
          }

\date{Received: date / Accepted: date}

\fi

\title{Equilibrium Computation in Resource Allocation Games}

\ifarxiv

\author[1]{Tobias Harks}
	\author[2]{Veerle Tan-Timmermans}
	\affil[1]{Department of Mathematics, Augsburg University}
	\affil[2]{Department of Quantitative Economics, Maastricht University}
	
\fi

\begin{document}
\maketitle
\begin{abstract}
We study the equilibrium computation problem for  two classical resource allocation games:  atomic splittable congestion games and multimarket Cournot oligopolies. For atomic splittable congestion games with singleton strategies and player-specific affine cost functions, we devise the first polynomial time algorithm computing a pure Nash equilibrium. Our algorithm is combinatorial and computes the \emph{exact} equilibrium assuming rational input.  The idea is to compute an equilibrium for an associated \emph{integrally-splittable}  singleton congestion game in which the players can only split their demands in integral multiples of a common packet size.  While integral games have been considered in the literature before, no polynomial time algorithm computing an equilibrium was known. Also for this class, we devise the first polynomial time algorithm and use it as a building block for our main algorithm.	
			
We then develop a polynomial time computable transformation mapping a multimarket Cournot competition game with firm-specific affine price functions and quadratic costs to an associated atomic splittable congestion game as described above. The transformation preserves equilibria in either game and, thus, leads -- via our first algorithm -- to a polynomial time algorithm computing Cournot equilibria. Finally, our analysis for integrally-splittable games implies new bounds on the difference between real and integral Cournot equilibria. The bounds can be seen as a generalization of the recent bounds for single market oligopolies obtained  by Todd~\cite{Todd16}.\footnote{An extended abstract of parts of this paper appeared in the Proceedings of the 19th International IPCO Conference on Integer Programming and Combinatorial Optimization under the title ``Equilibrium Computation in Atomic Splittable Singleton Congestion Games''.}

\ifllncs
\keywords{Atomic Splittable Congestion Games \and Multimarket Cournot Competition \and Equilibrium Computation}
 \subclass{91A10 \and 91A46 \and 91B32}
 \fi
\end{abstract}

\section{Introduction}
	\label{sec:introduction}
One of the core topics in computational economics, operations research and optimization
is the computation of equilibria. As pointed out by several researchers (e.g.~\cite{Chen09complexity,DaskalakisGP09}), the computational tractability of a solution concept contributes to its credibility as a plausible
prediction of the outcome of competitive environments in
practice. The most accepted solution concept in
non-cooperative game theory is the Nash equilibrium -- a strategy profile, from
which no player wants to unilaterally deviate. While a Nash equilibrium 
generally exists only in mixed strategies, the practically important class of congestion
games admits pure Nash equilibria, see Rosenthal~\cite{Rosenthal73a}. In the classical model of Rosenthal, a pure strategy of a player consists of a subset of resources, and the congestion cost of a resource depends only on the number of players choosing the same resource. 

While the complexity of computing equilibria for (discrete) congestion games has been intensively studied over the last decade (cf.~\cite{Ackermann08,CaragiannisFGS11,CaragiannisFGS15,Chien11,Fabrikant04,Skopalik08}),
the equilibrium computation problem for the \emph{continuous} variant, that is, for \emph{atomic splittable congestion games} is much less explored.
 In such a game, a player is associated with a positive demand and a
collection of allowable subsets of the resources. A strategy for a player is
a (possibly fractional) distribution of the player-specific demand over the allowable
subsets. This quite basic model has been extensively studied,
starting in the 80's in the context of traffic networks (Haurie and Marcotte~\cite{Haurie85})
and later for modeling communication networks (cf. Orda et al.~\cite{Orda93} and Korilis et al.~\cite{Korilis1997,KorilisLO95}), and logistics networks (Cominetti et al.~\cite{Cominetti09}).
Regarding polynomial time algorithms for equilibrium computation, we are only aware of four works: (1) For affine player-independent cost functions, there
exists a convex potential whose global minima are pure Nash equilibria, see~Cominetti et al.~\cite{Cominetti09}.
Thus, for any $\epsilon>0$  one can compute an $\epsilon$-approximate
equilibrium in polynomial time by convex programming methods. 
(2) Huang~\cite{Huang13} also considered affine player-independent cost functions, and he devised
a combinatorial algorithm computing an exact equilibrium for routing games
on symmetric $s$-$t$ graphs that are so-called \emph{well-designed}.
This condition is met for instance by series-parallel graphs. His proof technique
also uses the convex potential. 
(3) After the initial publication of the conference version
of this article, Bhaskar and Lolakapuri~\cite{BhaskarL18} proposed two algorithms with exponential worst-case complexity that compute approximate Nash equilibria in games with convex costs, when set systems consist of singletons only.
(4) Klimm and Warode~\cite{klimm18} recently proved that computing a pure Nash equilibrium
for atomic splittable and integer-splittable network congestion games
with affine player-specific costs is PPAD-complete (see~\cite{Papadimitriou94}). 
In light of these hardness results, it becomes clear that some
restrictions on the strategy space are likely to be necessary to obtain
polynomial time algorithms for equilibrium computation.

\subsection{Our Results and Techniques}
\label{subsec:ourresults}

\paragraph{Atomic Splittable  Congestion Games.}
 We study atomic splittable congestion games as defined above, where the set systems consist of singletons only, and cost functions are player-specific, increasing and affine.
We call these games \emph{atomic splittable singleton congestion games} and 
for these games we develop the first polynomial time algorithm computing a pure Nash equilibrium. From now on we use equilibrium as shortcut for pure Nash equilibrium. Our algorithm is purely combinatorial and computes an \emph{exact} equilibrium.
The main ideas and constructions are as follows. By analyzing the first
order necessary optimality conditions of an equilibrium, it can be shown that any
equilibrium is \emph{rational} as it is a solution to a system of linear equations with rational
coefficients (assuming rational input). Using that equilibria are unique for singleton games (see Richmann and Shimkin~\cite{Richman07} and Bhaskar et al.~\cite{Bhaskar15}), we further derive that the constraint
matrix of the equation system is non-singular, allowing for an explicit
representation of the equilibrium by Cramer's rule (using determinants of the constraint- and their sub-matrices). This way, we obtain an explicit  lower bound
on the minimum demand value for any used resource in the equilibrium. 
We further show that the unique equilibrium is also
an equilibrium for an associated \emph{integrally-splittable} game in which the players may only distribute the demands in \emph{integer multiples} of a common \emph{packet size} of some value $k^*\in \Q_{>0}$ over the resources. Moreover, all equilibria in this integral splittable game are very similar. While we are not able to compute $k^*$ exactly, we can efficiently compute some sufficiently small $k_0\leq k^*$ with the property that an equilibrium for the $k_0$-integrally-splittable game allows us to determine the set of resources on which a player will put a positive amount of load in the atomic splittable equilibrium. Once these  \emph{support sets} are known, an atomic splittable equilibrium can be computed in polynomial time by solving a system of linear equations.
This way, we can reduce the problem of computing the exact equilibrium for an atomic splittable
game to computing an equilibrium for an associated $k_0$-integrally-splittable game.

The class of integrally-splittable congestion games has been studied before by Tran-Thanh et al.~\cite{Tran11} for the case of player-independent convex cost functions and later by Harks et al.~\cite{harks2016resource}  (for the more general
case of polymatroid strategy spaces and player-specific convex cost functions). In particular, Harks et al. devised an algorithm with running time
$n^2m(\delta / k_0)^3$, where $n$ is the number of players, $m$ the number of resources, and
$\delta$ is an upper bound on the maximum demand of the players (cf. Theorem 5.2~\cite{harks2016resource}).
As $\delta$ is encoded in binary, however, the algorithm is only pseudo-polynomial even for player-specific affine cost functions.

We devise a polynomial time algorithm for integrally-splittable singleton congestion games 
with player-specific affine cost functions. Our algorithm works as follows.
For a game with initial packet size $k_0$, we start by finding an equilibrium for packet size $k=k_0\cdot 2^q$ for some $q$ of order $O(\log(\delta/k_0))$, satisfying only a part of the player-specific demands. Then we repeat the following two actions:
\begin{enumerate}
	\item We halve the packet size from $k$ to $k/2$ and construct in polynomial time a $k/2$-equilibrium using the $k$-equilibrium. Here, a $k$-equilibrium denotes an equilibrium for an integrally-splittable game
	with common packet size $k$.  
		\item For each player $i$ we repeat the following step: if the current packet size $k$ is smaller than the currently unscheduled demand of player $i$, we add one more packet for this particular player to the game and recompute the equilibrium. This part of the algorithm has also been used in the algorithm by Tran-Thanh et al.~\cite{Tran11} and Harks et al.~\cite{harks2016resource}.
		\end{enumerate}
After $q$ iterations, we have scheduled all demands and obtain an equilibrium for the desired packet size $k_0$.

Key to the analysis of the correctness and the running time of the
algorithm are several structural
results on the sensitivity of equilibria with respect to different integral packet sizes $k\in  \Q_{>0}$ and $k/r\in \Q_{>0}$ for some $r\in \N$.
Specifically, we derive bounds on the difference of resulting global  load vectors as well
as  individual load vectors of players in any respective equilibrium. 
These sensitivity results may be of independent interest
as they  show how equilibria gradually behave in terms of the
discretization  granularity.

Overall, compared to the existing algorithms of Tran-Thanh et al.~\cite{Tran11} and Harks et al.~\cite{harks2016resource},  our algorithm has two main innovations: packet sizes are decreased exponentially (yielding polynomial running time in $\delta$) and $k$-equilibrium computation
for an intermediate packet size $k$ is achieved in polynomial running time. 

\paragraph{Multimarket Cournot Oligopolies.}
We then study the equilibrium computation problem for Cournot oligopolies. In the basic model of Cournot~\cite{cournot1838recherches} introduced in 1838, firms produce homogeneous goods and sell them in a \emph{common} market.
The selling price of the goods depends on the total amount of goods that is offered in the market. Each firm aims to maximize its profit, which is equal to the revenue minus the production costs.
In a \emph{multimarket oligopoly} (cf. Bulow~\cite{Bulow85}), firms compete over a \emph{set} of markets and each firm has access to a firm-specific subset of the markets. 

For multimarket oligopolies, we develop a poly-time computable isomorphism mapping a multimarket Cournot competition game to an associated
		atomic splittable singleton congestion game. The isomorphism is
		payoff invariant (up to constants) and thus  preserves equilibria in either games. 
		As a consequence, we can apply the isomorphism and the polynomial time algorithm 
		for atomic splittable congestion games to efficiently compute Cournot equilibria 
		for models with firm-specific affine price functions and quadratic production costs. In addition,
		 our analysis for integrally-splittable games also implies 
		 new bounds on the difference between real and integral Cournot equilibria
		 complementing and extending recent results of Todd~\cite{Todd16}. The case of affine price functions
		 with quadratic cost functions is a well-studied 
		 model in economics, see Moulin et al.~\cite{MoulinRG14}
and further references therein.

\subsection{Related Work}
\label{subsec:relatedwork}
\paragraph{Discrete Congestion Games.}
As the first seminal work regarding the computational complexity of equilibrium computation
in congestion games, Fabrikant et al.~\cite{Fabrikant04} showed that the problem of computing a pure Nash equilibrium is PLS-complete for network congestion games. Ackermann et al.~\cite{Ackermann08} strengthened this result to hold even for network congestion games with linear cost functions. 
On the other hand, there are polynomial algorithms for symmetric network congestion games (cf. Fabrikant et al.~\cite{Fabrikant04}), for matroid congestion games with player-specific cost functions (Ackermann et al.~\cite{Ackermann08,Ackermann09}) and for so-called total unimodular congestion games (see Del Pia et al.~\cite{Delpia17}).

In particular, there is a pseudo-polynomial time algorithm that computes pure Nash equilibria for polymatroid congestion games with player-specific cost functions and polynomially bounded demands (Harks et al.~\cite{harks2016resource}). As mentioned in Section~\ref{subsec:ourresults}, their results plays a significant role in this paper. The algorithm by Harks et al. starts with the trivial equilibrium for the game where all player-specific demands are zero. Then, they sequentially add packets to the game. After a packet is added, additional packet exchanges might be executed to recompute the equilibrium. For the special case of affine cost functions and singleton strategy spaces we construct an alternative algorithm that can compute equilibria in polynomial time.

Further results regarding the computation of approximate equilibria in congestion games can be found in Caragiannis et al.~\cite{CaragiannisFGS11,CaragiannisFGS15}, Chien and Sinclair~\cite{Chien11} and Skopalik and V\"ocking~\cite{Skopalik08}.

\paragraph{Atomic Splittable Congestion Games.}
Atomic splittable congestion games on networks with player-independent cost functions have been studied (seemingly independently) by Orda et al.~\cite{Orda93} and Haurie and Marcotte~\cite{Haurie85} and Marcotte~\cite{Marcotte1987}.
Both lines of research mentioned that Rosen's existence result for concave games
on compact strategy spaces implies the existence of pure Nash equilibria via Kakutani's fixed-point theorem. Cominetti et al.~\cite{Cominetti09} presented the first upper bounds on the price
of anarchy in atomic splittable congestion games. These were later improved by Harks~\cite{Harks:stack2011} and finally shown to be tight by Schoppmann and Roughgarden~\cite{Roughgarden15}.

For the computation of equilibria, Marcotte~\cite{Marcotte1987} proposed four numerical algorithms and  showed local convergence results. 
Meunier and Pradeau~\cite{MeunierP13} developed a pivoting-algorithm (similar to Lemke's algorithm)
for nonatomic network congestion games with affine player-specific cost functions. Polynomial running time was, however, not shown and seems unlikely to hold.
Gairing et al.~\cite{Gairing11} considered nonatomic routing games on parallel links
with affine player-specific cost functions. They developed a convex potential function
that  can be minimized within arbitrary precision in polynomial time.
Deligkas et al.~\cite{DeligkasFS16} considered general concave games with
compact action spaces and investigated algorithms computing an approximate equilibrium.
Roughly speaking, they discretized the compact strategy space and use
the Lipschitz constants of utility functions to show that only a finite number
of representative strategy profiles need to be considered for obtaining an approximate
equilibrium (see also Lipton et al.~\cite{LiptonMM03} for a similar approach). The
running time of the algorithm, however, depends on an upper bound of the norm of strategy vectors, thus, implying only a pseudo-polynomial algorithm for our setting. 

Note that the problem of computing pure Nash equilibria in atomic splittable congestion games with singleton strategies and affine cost functions can be written as a \emph{linear complementary problem}, but does not seem to fall in any of the classes for which a solution can be found in polynomial time.

\paragraph{Multimarket Cournot Oligopolies.}
The existence of equilibria in single market Cournot models (beyond quasi-polynomial utility functions) has been studied extensively in the past decades (see Vives~\cite{Vives05} for a good survey). E.g., Novshek~\cite{Novshek85} proved that equilibria exists whenever the marginal revenue of each firm is decreasing in the aggregate quantities of the other firms. Then, several works (cf. Topkis~\cite{Topkis79}, Amir~\cite{Amir96}, Kukushkin~\cite{Kukushkin94}, Milgrom and Roberts~\cite{Milgrom90}, Milgrom and Shannon~\cite{Milgrom94}, Topkis~\cite{Topkis98} and Vives~\cite{Vives90}) proved existence of equilibria when the underlying game is supermodular, i.e., when the strategy space forms a lattice and the marginal utility of each firm is increasing in any other firm's output. Using supermodularity, one can obtain existence results without assuming that the utility functions are quasi-convex. Very recently, Todd~\cite{Todd16} considered Cournot competition on a single market, where the price functions are linear and cost functions are quadratic. For such games, he proved that equilibria exist and can be computed in time $O(n\log(n))$, where $n$ denotes the number of firms. Additionally, he analyzed the maximum differences of production quantities of real and integral equilibria, respectively.

Abolhassani et al.~\cite{Abolhassani14} devised several polynomial time
algorithms for multimarket Cournot oligopolies, partly using algorithms for solving nonlinear complementarity problems. In contrast to our work, they  assume that
price functions are firm-independent.
 Bimpikis et al.~\cite{Bimpikis14} provided a characterization of the production quantities at the unique equilibrium, when price functions are player-independent and concave, and cost functions are convex. They study the impact of changes in the competition structure on the firm's profit. This framework can be used to either identify opportunities for collaboration and expanding in new markets.
Harks and Klimm~\cite{Harks14} studied the existence of Cournot equilibria, under the condition that each firm can only sell its items to a limited number of markets simultaneously. They  proved that equilibria exist when production cost functions are convex, marginal return functions strictly decrease for strictly increased own quantities and non-decreased aggregated quantities and when for every firm, the firm specific market reaction functions across markets are identical up to market-specific shifts.

\section{Preliminaries}
\label{sec:preliminaries}
\paragraph{Atomic Splittable Singleton Games.}
\label{sec:gamedefinition}

An atomic splittable singleton congestion game is defined by a tuple: $\G:= \left(N,E,(d_i)_{i \in N}, (E_i)_{i \in N},(c_{i,e})_{i \in N, e \in E_i} \right),$ where $E=\{e_1,\dots,e_m\}$ is a finite set of resources and
$N=\{1,\dots,n\}$ is a finite set of players. 
Each player $i\in N$ is associated with a demand $d_i \in \Q_{\geq 0}$ and a set of allowable resources $E_i \subseteq E$. A strategy for player $i\in N$ is a (possibly fractional) distribution of the demand $d_i$
over the singletons in $E_i$. 
Thus, one can represent the strategy space of every player $i\in N$ by the polytope: 
$$\S_i(d_i) := \{  x_i \in \mathbb{\R}_{\geq 0}^{|E_i|} \mid \sum_{e \in E_i} x_{i,e} = d_i \}.$$

The combined strategy space is denoted by $\S := \bigtimes _{i \in N} \S_i(d_i)$ and ${x}=({x}_i)_{i\in N}$ is the overall strategy profile. We define $x_{i,e}:=({x}_i)_e$ as the load of player $i$ on $e \in E_i$ and $x_{i,e} = 0$ when $e \in E \setminus E_i$. The total load on resource $e$ is given as $x_e:=\sum_{i\in N}x_{i,e}$. The total load on resource $e$ minus the contribution of
player $i$ is given by $x_{-i,e}:=\sum_{j\in N\setminus\{i\}}x_{j,e}$.

 Resources have player-specific affine cost functions $c_{i,e}(x_e) = a_{i,e}x_e + b_{i,e}$ with $a_{i,e} \in \Q_{> 0}$ and $b_{i,e} \in \Q_{\geq 0}$ for all $i \in N$ and $e \in E$. The total cost of player $i$ in strategy distribution ${x}$ is defined as: 
$\pi_{i}(x)=\sum_{e\in E_i} c_{i,e}(x_e)\,x_{i,e}.$
We
write $\S_{-i}(d_{-i}) = \bigtimes_{j \neq i} \S_{j}(d_j)$ and we write $x = (x_i,x_{-i})$ for each $i \in N$,
meaning that $x_i \in \S_i(d_i)$ and $x_{-i} \in \S_{-i}(d_{-i})$. A strategy profile
$x$ is an \emph{equilibrium} if $\pi_i(x) \leq
\pi_i(y_i, x_{-i})$ for all $i \in N$ and $y_i \in \S_i(d_i)$. A pair
$\bigl(x,(y_i,x_{-i})\bigr) \in \S\times \S$ is called an
\emph{improving move} of player $i$,
if $\pi_i(x_i,x_{-i}) > \pi_i(y_i,x_{-i})$. We define $\mu_{i,e} ( x) = c_{i,e}(x_e) + x_{i,e}c'_{i,e}(x_e) = a_{i,e}(x_e + x_{i,e}) + b_{i,e}$ to be the \emph{marginal cost} for player $i$ on resource $e$.
We obtain the following sufficient and necessary equilibrium condition.

\begin{lemma}[cf. Harks~\cite{Harks:stack2011}]
	\label{equilibriumcondition}
	Strategy profile $x$ is an equilibrium if and only if the following holds for all $i\in N$:  if $x_{i,e}>0$, then $\mu_{i,e}( x)\leq \mu_{i,f}( x)$ for all $f\in E_i$. 
\end{lemma}
Using that the strategy space is compact and cost functions are convex, Kakutani's fixed point theorem implies the existence of an equilibrium. Uniqueness is proven by Richmann and Shimkin~\cite{Richman07} and Bhaskar et al.~\cite{Bhaskar15}. 

Game $\G$ is called symmetric whenever $E_i = E$ for all $i \in N$.
We can project any asymmetric game $\G$ on a symmetric game $\G^*$ by  setting
$c^*_{i,e}(x_e)$ to $c_{i,e}(x_e)$ whenever $e \in E_i$, and to $x_e + (n+2)(a_{\max})^2$ otherwise. Here, $$a_{\max} :=\max\{\{a_{i,e}, b_{i,e} \mid i \in N, e\in E_i \},\{d_i \mid i \in N \},1\}.$$ 
In this case $\mu_{i,e}(0) \geq \mu_{i,f}(x_e)$ for any $e \in E \setminus E_i$, $f\in E_i$, $i \in N$ and $x \in \S$. Thus, in an equilibrium $y$ for game $\G^*$ no player $i$ puts load on any resource $e \in E \setminus E_i$. Hence, $y$ is also an equilibrium for game $\G$. In the rest of this paper, we project every asymmetric game on a symmetric game using the construction above and write  $\G:= \left(N,E,(d_i)_{i \in N},(c_{i,e})_{i \in N, e \in E_i} \right)$ instead.

\paragraph{Integral Singleton Games.} 
A $k$-integral singleton game is compactly defined by the tuple
$\G_k:= \left(N,E,(d_i)_{i \in N}, (c_{i,e})_{i \in N, e \in E} \right)$ with $k\in \Q_{>0}$. Here, players cannot split their load fractionally, but only in multiples of $k$. Assume $d_i$ is a multiple of $k$, then the strategy space for player $i$ is the following set: 
$$\S_i(d_i,k) := \left\{  x_i \in \mathbb{\Q}_{\geq 0}^{|E|} \mid x_{i,e} = kq_{i,e}, q_{i,e}\in \N_{\geq 0}, \textstyle \sum_{e \in E} x_{i,e} = d_i \right\}.$$
In this game, $k$ is also called the \emph{packet size}. When $k$ and $d_i$ are clear from the context, we refer to $\S_i(d_i,k)$ as $\S_i.$  When $E, N$ and $(c_{i,e})_{i \in N, e \in E}$ are clear from the context, we also refer to the game as $\G_k((d_i)_{i \in N}).$ For player-specific affine cost functions, the (discrete) marginal costs are defined as follows:
\begin{align} \label{marginalcostup}
\mu_{i,e}^{+k} (x) &= (x_{i,e}+k)c_{i,e}(x_e+k) - x_{i,e}c_{i,e}(x_e), \\
\mu_{i,e}^{-k} (x) &=
\begin{cases}
x_{i,e}c_{i,e}(x_e) -  (x_{i,e}-k)c_{i,e}(x_e-k),
& \text{if } x_{i,e} >  0\\
-\infty,
& \text{if } x_{i,e}\leq 0.
\end{cases}
\end{align}	
Here, $\mu_{i,e}^{+k} (x)$ represent  the cost increment for player $i$, if one packet of size $k$ is added to resource $e$ and $\mu_{i,e}^{-k} (x)$ denotes the cost saving for player $i$, if  one packet from resource $e$ is removed. Assuming cost functions are affine, we obtain $\mu_{i,e}^{+k} (x)= k(a_{i,e}(x_e + x_{i,e} + k)+b_{i,e})$ and  $\mu_{i,e}^{-k} (x) = k(a_{i,e}(x_e + x_{i,e} - k)+b_{i,e})$, if $x_{i,e}>0$. 

\begin{lemma}[cf. Groenevelt~\cite{groenevelt91}] \label{groenevelt}
	Strategy profile $x$ is an equilibrium in a $k$-integral congestion game if and only if
	for all $i\in N$ it holds that if $x_{i,e}>0$, then also $\mu^{-k}_{i,e}(x)\leq \mu^{+k}_{i,f}(x)$ for all $f\in E$. 
\end{lemma}

We also introduce some new notation. For two vectors $x_i, y_i \in \R^{|E|}$, we define $$H(x_i,y_i):=\sum_{e \in E} |x_{i,e} - y_{i,e}|$$ to be their Hamming distance. For two strategies  $x,y$, we write $H(x,y):=\sum_{e \in E} |x_{e} - y_{e}|$. 
We define a \emph{restricted improving move} and a \emph{restricted best response} as follows:
\begin{definition}
	Let $x$ be a strategy profile for game $\G_k((d_i)_{i \in N})$. 
	\begin{enumerate}
	\item A strategy $x'_i$ is called a \emph{restricted improving move} to $x$ for player $i$, if $$x'_i \in \{y_i \in \S_i(d_i,k) \mid H(x_i,y_i)=2k \text{ and } \pi_i(y_i, x_{-i}) < \pi_i(x_i, x_{-i}) \}.$$
	\item A strategy $x'_i$ is called a \emph{restricted best response} to $x$ for player $i$, if $$x'_i \in \arg\min_{y_i \in\S_i(d_i,k)}\{\pi_i(y_i, x_{-i}) \mid H(x_i,y_i)=2k\}.$$
	 \end{enumerate}
\end{definition}
Note that both a restricted improving move and a restricted best response can be executed by moving a single packet.

\section{Sensitivity Results for Equilibria}
In Section~\ref{sec:algorithm}, we show that computing an equilibrium
for atomic splittable games can be reduced to the problem of computing an equilibrium of an associated \emph{integrally-splittable}
game with small enough packet size.
For such a class of discrete games, we will develop 
a polynomial time \emph{scaling algorithm}, where
we write the total demand as a power of two 
and then iteratively scale down the allowed packet size
and recompute equilibria for the resulting integrally-splittable games.
The key for the well-definedness and further analysis of this algorithm
is a structural result on the sensitivity of equilibria
for integrally-splittable games with respect
to changed packet sizes.
In the following, we derive such sensitivity results between equilibria of an integrally-splittable game $\G_k$ with packet size $k\in\Q_{>0}$
and those of a game $\G_{\frac{k}{r}}$ with $r\in\N$.
These results may be of independent
interest in the area of comparative statics, where
the influence of parameters w.r.t. to resulting equilibria are analyzed.

\begin{theorem}\label{thm:globalload2}
	Let $x_k$ be an equilibrium for game $\G_k$, and $x_{k/r}$ be an equilibrium for game $\G_{k/r}$. Then $|(x_k)_e-(x_{k/r})_e| \leq n(m-1)(1 + \tfrac{1}{r})k$ for all $e \in E$.
\end{theorem}
\begin{proof}
Let us first recall the definition of marginal cost functions for players with packet sizes
$k/r$.
\begin{equation}\label{eq:marginal}
\begin{aligned}
\mu^{+k/r}_{i,e}(x_{k/r})&
= \tfrac{k}{r} (a_{i,e}((x_{k/r})_{i,e} + (x_{k/r})_e+\tfrac{k}{r})+ b_{i,e}).\\
\mu_{i,e}^{-k/r} (x_{k/r}) &= \begin{cases}
 \tfrac{k}{r} (a_{i,e}((x_{k/r})_{i,e} + (x_{k/r})_e-\tfrac{k}{r})+ b_{i,e}), & \text{if } (x_{k/r})_{i,e} >  0\\
-\infty, & \text{if } (x_{k/r})_{i,e}\leq 0.
\end{cases}
\end{aligned}
\end{equation}

	In order to prove the theorem we need to show that both:
	\begin{enumerate}
	\item $(x_k)_e-(x_{k/r})_e  \leq  n(m-1)(1 + \tfrac{1}{r})k$ and 
	\item $(x_{k/r})_e  - (x_k)_e \leq  n(m-1)(1 + \tfrac{1}{r})k.$
	\end{enumerate}
	 As the proofs for both statements are very similar, we only prove the first statement here. On the contrary, assume that there exists a resource $e_1$ with 
	 \[ (x_k)_{e_1}-(x_{k/r})_{e_1}  >  n(m-1)(1 + \tfrac{1}{r})k.\] We introduce two edge sets $E^+,E^-$ as: 
	\[ 
	E^+ = \{e \in E |   (x_{k})_{e} \geq (x_{k/r})_{e} \} \text{ and } E^- = \{e\in E | (x_{k})_{e} < (x_{k/r})_{e}  \}.
	\] 
	We get 
\[	
	\sum_{e \in E^+} ( (x_{k})_{e} - (x_{k/r})_{e}) =\sum_{i\in N}\sum_{e \in E^+} ( (x_{k})_{i,e} - (x_{k/r})_{i,e})  > n(m-1)(1 + \tfrac{1}{r})k.\]
	Thus, with $n=|N|$, there exists $i\in N$ with 
	\begin{equation}\label{eq:i_project_O} \sum_{e \in E^+} ( (x_{k})_{i,e} - (x_{k/r})_{i,e})  > (m-1)(1 + \tfrac{1}{r}) k.\end{equation}
	With $|E^+|\leq m-1$, there exists $e\in E^+$ with 
\[	(x_{k})_{i,e} - (x_{k/r})_{i,e}  > (1 + \tfrac{1}{r}) k.\]
With 	\eqref{eq:i_project_O} and the balance constraint $\sum_{e\in E}(x_{k})_{i,e} - (x_{k/r})_{i,e}=0$,
we get
\[ \sum_{e \in E^-} ( (x_{k})_{i,e} - (x_{k/r})_{i,e})  < - (m-1)(1 + \tfrac{1}{r}) k.\]
Again using $|E^-|\leq m-1$, there is $f\in E^-$ with 
\begin{equation}\label{eq:ikf}(x_{k})_{i,f} - (x_{k/r})_{i,f}  < - (1 + \tfrac{1}{r}) k.\end{equation}
Note that~\eqref{eq:ikf} implies $(x_{k/r})_{i,f}>0$.
Altogether we have for some $i\in N, e\in E^+, f\in E^-$ the following conditions:
\begin{align}
(x_{k})_{e} - (x_{k/r})_{e}+ (x_{k})_{i,e} - (x_{k/r})_{i,e}  > (1 + \tfrac{1}{r}) k \label{eq:over}\\
(x_{k})_{f} - (x_{k/r})_{f}+ (x_{k})_{i,f} - (x_{k/r})_{i,f}  < - (1 + \tfrac{1}{r}) k\label{eq:under}.
\end{align}
We  rearrange Equations~\eqref{eq:over}, \eqref{eq:under} by first multiplying  the respective inequalities with $a_{i,e}, a_{i,f}$ and then adding $b_{i,e}, b_{i,f}$ to the respective sides to obtain
\begin{align*}
a_{i,e}((x_{k})_{e} + (x_{k})_{i,e} -k) +b_{i,e}  & >  a_{i,e}((x_{k/r})_{e}+(x_{k/r})_{i,e} + \tfrac{k}{r}) 	+b_{i,e}\\
a_{i,f}((x_{k})_{f} + (x_{k})_{i,f} +k) +b_{i,f} &<  a_{i,f}((x_{k/r})_{f}+(x_{k/r})_{i,f} - \tfrac{k}{r}) +b_{i,f}.
\end{align*}
We multiply both inequalities with $\tfrac{k}{r}>0$ and obtain
\begin{align}\label{eq:over2}
\tfrac{k}{r}(a_{i,e}((x_{k})_{e} + (x_{k})_{i,e} -k) +b_{i,e} )  &>  \tfrac{k}{r}(a_{i,e}((x_{k/r})_{e}+(x_{k/r})_{i,e} + \tfrac{k}{r}) 	+b_{i,e})\\\label{eq:under2}
\tfrac{k}{r}(a_{i,f}((x_{k})_{f} + (x_{k})_{i,f} +k) +b_{i,f})&<  \tfrac{k}{r}(a_{i,f}((x_{k/r})_{f}+(x_{k/r})_{i,f} - \tfrac{k}{r}) +b_{i,f}).
\end{align}
	We combine Equation~\eqref{eq:over2}, Equation~\eqref{eq:under2} and the fact that $x_k$ is an equilibrium for packet size $k$ to obtain (recall~\eqref{eq:marginal} and $(x_{k})_{i,e}>0$ and $(x_{k/r})_{i,f}>0$ ): 
	\begin{equation*}
	\mu^{+k/r}_{i,e}(x_{k/r})
	\underset{\eqref{eq:over2}}{<} \frac{1}{r}\mu^{-k}_{i,e}(x_{k})  
	\leq \frac{1}{r} \mu^{+k}_{i,f}(x_{k})  
	\underset{\eqref{eq:under2}}{<} \mu^{-k/r}_{i,f}(x_{k/r}).  
	\end{equation*}
	Hence,  player $i$ that has a restricted improving move in $x_{k/r}$ shifting a packet of size $k/r$ from $f$ to $e$, which contradicts the fact that $x_{k/r}$ is an equilibrium strategy. \ifllncs \qed \fi
\end{proof}

With 
 \[ 
 	\lim_{r \rightarrow \infty} \frac{r}{k} \mu^{+k/r}_{i,e}(x) = 	\lim_{r \rightarrow \infty} \frac{r}{k} \mu^{-r/k}_{i,e}(x) = \mu_{i,e}(x), 
 \]
 we immediately obtain the following statement from Theorem~\ref{thm:globalload2}.

\begin{corollary}\label{cglobalload2}
	Let $x$ be the unique equilibrium for an atomic splittable game, and $x_k$ be an equilibrium for a $k$-integral splittable game. Then $|(x_k)_e-x_e| \leq n(m-1)k$ for all $e \in E$.
\end{corollary}


We obtain a similar result for player-specific load differences:

\begin{theorem}\label{thm:localload2}
	Let $x_k$ be an equilibrium for game $\G_k$, and $x_{k/r}$ be an equilibrium for game $\G_{k/r}$. Then $|(x_k)_{i,e}-(x_{k/r})_{i,e}| \leq m n (m-1)(1 + \tfrac{1}{r})k$ for all $e \in E, i\in N$.
\end{theorem}

\begin{proof}
In order to prove the theorem we need to show that for all $i\in N, e\in E$ the following two inequalities hold:
	\begin{enumerate}
	\item $(x_k)_{i,e}-(x_{k/r})_{i,e} \leq  m n(m-1)(1 + \tfrac{1}{r})k$ and 
	\item $(x_{k/r})_{i,e}-(x_k)_{i,e} \leq  m n(m-1)(1 + \tfrac{1}{r})k$
	\end{enumerate}
	 We again only prove the first statement here. Assume
	 by contradiction that there exists a resource $e_1$ and a player $i\in N$ with  $(x_k)_{i,e_1}-(x_{k/r})_{i,e_1} >  m n(m-1)(1 + \tfrac{1}{r})k$. By Theorem~\ref{thm:globalload2}, we know that  $(x_{k})_{e_1}-(x_{k/r})_{e_1}  \geq - n (m-1)(1 + \tfrac{1}{r})k $. Adding both inequalities, we get:	 
	 \[(x_{k})_{e_1}-(x_{k/r})_{e_1}+(x_k)_{i,e_1}-(x_{k/r})_{i,e_1}>n(m-1)^2(1 + \tfrac{1}{r})k.\]
	 The total load distributed by all player is the same in both $x_{k/r}$ and $x_{k}$. Thus, we obtain: 
	\[
	\textstyle \sum_{e \neq e_1} (x_{k})_{e}-(x_{k/r})_{e}+(x_k)_{i,e}-(x_{k/r})_{i,e}  < - n(m-1)^2(1 + \tfrac{1}{r})k. 
	\]
	By the pigeonhole principle, there must exist at least one resource $f \in E, f\neq e_1$ such that:
	\begin{equation} 
	\label{ceq31} (x_{k})_{f}-(x_{k/r})_{f}+(x_k)_{i,f}-(x_{k/r})_{i,f}  < - n(m-1)(1 + \tfrac{1}{r})k.
	\end{equation}
	Note that $(x_{k/r})_{i,f} > 0$, as $(x_{k/r})_{i,f} = 0$ implies $(x_{k/r})_f - (x_{k})_f > n(m-1)(1 + \tfrac{1}{r})k$, which contradicts Theorem~\ref{thm:globalload2}. We obtain in a similar fashion as in  the proof of Theorem~\ref{thm:globalload2}:
	
\begin{equation}\label{b3}
	\mu^{+k/r}_{i,e_1}(x_{k/r})
	< \frac{1}{r}\mu^{-k}_{i,e_1}(x_{k})  
	\leq \frac{1}{r} \mu^{+k}_{i,f}(x_{k})  
	\leq \mu^{-k/r}_{i,f}(x_{k/r}).  
	\end{equation}
	
	As  $(x_{k/r})_{i,f} > 0$, player $i$ has a restricted improving move from resource $f$ to resource $e_1$ contradicting the fact that  $x_{k/r}$ is an equilibrium strategy.  
	\end{proof}

 Again, we immediately obtain the following statement from Theorem~\ref{thm:localload2}.

\begin{corollary}\label{clocalload2}
	Let $x$ be the unique equilibrium for an atomic splittable game, and $x_k$ be an equilibrium for the corresponding $k$-integral splittable game. Then $|(x_k)_{i,e}-x_{i,e}| \leq  nm(m-1)k$ for all $i \in N$ and $e\in E$.
\end{corollary}

To complement Theorem~\ref{thm:globalload2} and Theorem~\ref{thm:localload2}, we provide a lower bound example where  
\[
|(x_k)_{i,e}-(x_{k/r})_{i,e}| = |(x_k)_e-(x_{k/r})_e| = (m-1) \frac{k}{r}.
\]

\begin{example}
Consider a $k-$splittable congestion game $\G_k$ with player set $N=\{1\}$ and resource set $\{e_1, \dots, e_m\}$. Let $d_1 = (m-1)k$, and the cost functions are defined as follows:
\begin{align*}
c_{1,e}(x_e):=
\begin{cases}
\frac{x_e}{2(r-1)(m-1)}  & \mbox{  if } e = e_m, \\
x_e & \mbox{ otherwise}.\\
\end{cases}
\end{align*}
In game $\G_k$, a best response $x_k$ for player 1 is to put all $m-1$ packets on resource $e_m$. Alternatively, if the packet size is $\tfrac{k}{r}$ instead of $k$, strategy 
\[
x_{k/r}:=(\tfrac{k}{r}, \dots, \tfrac{k}{r}, (m-1)(k-\tfrac{k}{r})),
\]
is an equilibrium strategy for player 1.
\end{example}

\section{Reduction to Integrally-Splittable Games}
\label{sec:algorithm}
We show that the problem of finding an equilibrium for an atomic splittable game reduces to the problem of finding an equilibrium for a $k_0$-integral game for some $k_0\in\Q_{>0}$. 

\begin{theorem} \label{atomictosplittable}
	Let $x$ be the unique equilibrium of an atomic splittable singleton game $\G$. Then, there exists $k^*\in \Q_{>0}$ such that $x$ is an equilibrium for game $\G_{k^*}$.
\end{theorem}
\begin{proof}
	We define the support set $I_i:=\{e \in E \mid x_{i,e}>0\}$ for each player $i \in N$. Lemma~\ref{equilibriumcondition} implies that if $x$ is an equilibrium, and $x_{i,e},x_{i,f}>0$, then $\mu_{i,e}(x) = \mu_{i,f}(x)$. 
	Define $p:= \sum_{i \in N} |I_i| \leq nm$. Then,  if the correct  support set $I_i$ of each player is known, the equilibrium can be computed by solving the following set of $p$ linear equations on $p$ variables. 
	\begin{enumerate}
		\item For every player we have an equation that makes sure the demand of that player is satisfied. Thus, for each player $i \in N$ we have $\sum_{e \in I_i} x_{i,e}= d_i$.
		\item  For every player $i \in N$, there are $|I_i| - 1$ equations of type $\mu_{i,e}(x)= \mu_{i,f}(x)$ for $e,f \in I_i$, which we write as $a_{i,e}(x_e + x_{i,e}) - a_{i,f}(x_{f} + x_{i,f}) = b_{i,e} - b_{i,f}$. Note that $x_e$ is not an extra variable, but an abbreviation for $\sum_{i \in N} x_{i,e}$.
	\end{enumerate}
	
	We refer to this set of equalities as $Ax=b$, where $A$ is a $p \times p$ matrix. Note that as the equilibrium exists and is unique, matrix $A$ is non-singular. Then, using Cramer's Rule, the unique solution of this system is given by: $x_{i,e} =\det(A_{i,e}) / \det (A) = |\det(A_{i,e})| / |\det (A)|,$ where $A_{i,e}$ is the matrix formed by replacing the column that corresponds to value $x_{i,e}$ in $A$ by $b$. We define $Q:=\{\{a_{i,e}, b_{i,e} \mid i \in N, e\in E_i \} \cup \{d_i \mid i \in N \} \cup \{1\} \}$ as the set of input values and  $a_{\gcd}:=\max\{a \in \Q_{>0} \mid \forall q \in Q, \; \exists \ell \in \N \text{ such that } q=a\cdot\ell \}$ as the \emph{greatest common divisor} of $Q$.
Note that since  $ 1\in Q$, we have $a_{\gcd}=\frac{1}{h}$ for some $ h\in \N$.
Then, as all values in $A$ and $b$ depend on adding and subtracting values in $Q$, $|\det(A_{i,e})|$ is an integer multiple of $(a_{\gcd})^{p}$ and, hence, (using $a_{\gcd}=\frac{1}{h}$) an integer multiple of $a_{\gcd}^{nm}$. Thus, all player-specific loads are integer multiples of $a_{\gcd}^{nm}/|\det(A)|$ and, hence, if we define $k^* =a_{\gcd}^{nm}/(2 \cdot|\det(A)|)$, $x$ is an equilibrium for the $k^*$-integral splittable game. Note that we can compute $a_{\gcd}$ in running time $O(nm\log{a_{\max}})$. 
\ifllncs \qed \fi
\end{proof}
We do not know matrix $A$ beforehand, but we do know that $2a_{\max}$ is an upper bound on the values occurring in $A$. Using Hadamard's inequality we find that $|\det(A)|\leq (2a_{\max})^{nm} (nm)^{nm/2}$. Hence, we can find a lower bound of $k^*$: 

$$k^* \geq a_{\gcd}^{nm}/((2a_{\max})^{nm} (nm)^{nm/2}).$$

By Corollary~\ref{cglobalload2} and Corollary~\ref{clocalload2}, we know that for the unique atomic splittable equilibrium $x$ and any $k$-integral-splittable equilibrium $x_k$, there exist bounds on $|x_e - (x_k)_e|$ and $|x_{i,e} - (x_k)_{i,e}|$ in terms of $k$, $n$ and $m$. Thus, if we compute an equilibrium for a sufficiently small $k_0$, this $k_0$-integral-splittable equilibrium should be fairly similar to the unique atomic splittable equilibrium. Hence, it enables us to find the correct support sets.  Then, given the correct support set of each player, we can compute the exact atomic splittable equilibrium by solving the system $Ax=b$ as described earlier.
\begin{theorem}
	Given an atomic splittable congestion game $\G$ and an equilibrium $x_{k_0}$ for the $k_0$-splittable game $\G_{k_0}$, where $k_0 := a_{\gcd}^{nm} / (3nm(m-1) \lceil (2a_{\max})^{nm} (nm)^{nm/2} \rceil) $.
	We can compute in $O((nm)^3)$ the unique atomic splittable equilibrium $x$ for game $\G$.
\end{theorem}
\begin{proof}
	First note that all demands $d_i$ are integer multiples of $k_0$, as $d_i$ is an integer multiple of $a_{\gcd}$, and both $3nm(m-1)$ and $\lceil (2a_{\max})^{nm} (nm)^{nm/2} \rceil$ are integers.  Theorem~\ref{atomictosplittable} implies that there exists a $k^*$ such that the atomic splittable equilibrium is also an equilibrium for the $k^*$-integral splittable game. Hence, for all $i\in N$ and $e \in E$ we have that $x_{i,e} = z_{i,e} \cdot k^*$ for some $z_{i,e} \in \N_{\geq 0}$. In the following we show that there is a load-threshold $\tfrac{3}{2}nm(m-1)k_0$ that enables us to decide whether or not a resource receives any demand from player $i$ in the equilibrium of the atomic splittable game.
		\begin{enumerate}
		\item If $(x_{k_0})_{i,e} < \tfrac{3}{2}nm(m-1)k_0$, then $x_{i,e} = 0$. Assume by contradiction that  $x_{i,e}>0$. Remember that the atomic splittable equilibrium is also a $k^*$-equilibrium and thus, if $x_{i,e}>0$, then the inequality $x_{i,e}\geq k^*$ must hold. We obtain $x_{i,e}-(x_{k_0})_{i,e} > k^* - \tfrac{3}{2}nm(m-1)k_0  \geq \tfrac{3}{2}nm(m-1)k_0, $
		where we used $k^*\geq 3nm(m-1)k_0$.
		This contradicts Corollary~\ref{clocalload2} and hence $x_{i,e} = 0$.
		\item If $(x_{k_0})_{i,e} \geq \tfrac{3}{2}nm(m-1)k_0$, then we prove that $x_{i,e} > 0$. Assume by contradiction that  $x_{i,e}=0$. We get $(x_{k_0})_{i,e}-x_{i,e} \geq \tfrac{3}{2} nm(m-1)k_0$, which contradicts Corollary~\ref{clocalload2}. Thus, $x_{i,e} > 0$.
	\end{enumerate}
Hence, given an equilibrium $(x_{k_0})$ for $k_0$-splittable game $\G_{k_0}$, we can compute the correct support sets $I_i= \{e \in E \mid (x_{k_0})_{i,e} \geq nm(m-1)k_0 \}$ for all $i \in N$. Given the correct support sets, we can easily compute the correct, exact equilibrium by solving the system $Ax=b$ of at most $nm$ linear equations in running time $O((nm)^3)$ using Gaussian elimination~\cite{Nemhauser1988}.  \ifllncs \qed \fi
\end{proof}
It is left to compute an equilibrium $x_{k_0}$ for the $k_0$-splittable game $\G_{k_0}$.
\section{A Polynomial Algorithm for Integral Games}
\label{sec:splittablealgorithm}
The goal of this section is to develop a \emph{polynomial time} algorithm that computes an equilibrium for any $k$-integral splittable singleton game with player-specific affine cost functions. 
Our algorithm uses as a subroutine the  algorithm of Harks, Klimm and Peis~\cite[Algorithm 1]{harks2016resource} that computes  in pseudo-polynomial time an equilibrium for any integer-splittable
congestion game with player-specific discrete-convex cost functions.

\subsection{The Algorithm of  Harks, Klimm and Peis~\cite[Algorithm 1]{harks2016resource}}

The algorithm of  Harks, Klimm and Peis~\cite[Algorithm 1]{harks2016resource}, which we denote by \textsc{PAlg}, starts with the empty strategy profile
and then inductively increases the demand of some player by one packet (of size $k$ in our case).
This new packet is placed greedily on some resource in order to minimize the private cost of the
respective player. By induction, the initial strategy profile (without the additional packet) is
an equilibrium and it is shown that by greedily placing the packet, the respective player plays a best response for the enlarged game.
Moreover, only players using the resource with the increased load by the new packet
may have an incentive to deviate, and, a best response consists of shifting a single packet away, see Fig.~\ref{alg:greedy}.
This leads to a structured restricted best response dynamic with a total running time of
 $O(n^2m(\frac{\delta }{k})^3)$, where $\delta:=\max_{i\in N} d_i$. Note that $\delta$ is not polynomially bounded in the input size,  thus, the running time is only pseudo-polynomial, hence, the naming \textsc{PAlg}.

\begin{algorithm}[tb]
 \Indm\Indmm
    \KwIn{$\G_{k}((d_i)_{i \in N})$}
  \KwOut{pure Nash equilibrium $ x$}
  \Indp\Indpp
$ d'_i \leftarrow 0$ and $ x_i\leftarrow  0$ for all $i\in N$\;
		\While{$ \sum_{i\in N}d'_i < \sum_{i\in N}d_i$}{	
			Choose $i\in N$ with $ d'_i<d_i$\;\label{it:choose_player}
			$d'_i\leftarrow d'_i+k$ \;\label{it:set_increase_demand}
	Choose a best response $y_i \in  \S_i(d'_i,k)$ with $H(y_i,x_i) = k$\;\label{it:comp_demand}
$ x_i \leftarrow  y_i$\; \label{it:increase_demand}
\While{$\exists i\in N$ who can improve in $\G_{k}((d'_i)_{i \in N})$
\label{it:if}}
 {Compute a best response $y_i \in \S_i(d'_i,k)$ with $H(y_i,x_i)=2k$\;  \label{it:choose_yi}
$ x_i \leftarrow  y_i$\;}
\label{it:endif}}
Return $ x$\;
 \caption{Pseudo-polynomial algorithm \textsc{PAlg} computing a pure Nash equilibrium.}
 \label{alg:greedy}
\end{algorithm}

\subsection{A Polynomial Time Algorithm}
We will construct a new algorithm named \PH\ with polynomial running time of order $O(n^5m^{10} \log(\delta/k))$.  
This algorithm works as follows.
For a game with desired packet size $k_0$, we start by finding an equilibrium for packet size $k=k_0\cdot 2^q$ for some $q$ of order $O(\log(\delta/k_0))$, satisfying only a part of the player-specific demands. Then we repeat the following two steps:
\begin{enumerate}
	\item {\bf Subroutine $\RESTORE$.} We half the packet size from $k$ to $k/2$ and construct a $k/2$-equilibrium using the $k$-equilibrium. Here, a $k$-equilibrium denotes an equilibrium for an integrally-splittable game with common packet size $k$.  
\item {\bf Subroutine $\ADD$.} For each player $i$ we repeat the following step: if the current packet size $k$ is smaller than the currently unscheduled demand of player $i$, we add one more packet for this particular player to the game and recompute the equilibrium. This part of the algorithm has also been used in the algorithm by Tran-Thanh et al.~\cite{Tran11} and Harks et al.~\cite{harks2016resource}.
\end{enumerate} 
After $q$ iterations, we have scheduled all demands and obtain an equilibrium for the desired packet size $k_0$. 
Let us now describe the two subroutines $\ADD$ and $\RESTORE$ and the main algorithm \PH\  in more detail below.

\subsubsection{\ADD}
The first subroutine, $\ADD$, is described in Algorithm~\ref{addfunction} and consists of lines 4-10 of~\cite[Algorithm 1]{harks2016resource}. Given an equilibrium $x_k$ for game  $\G_k((d_i)_{i \in N})$, it computes an equilibrium for the game, where the demand for some player $j$ is increased by a packet of size $k$. First it decides on the best resource $f$ for player $j$ to put her new packet. In effect, the load on resource $f$ increases and only those players with $x_{i,f}>0$ can potentially decrease their cost by a deviation. In this case, Harks et al. proved in~\cite[Theorem 3.2]{harks2016resource} that a best response $y_i$ can be obtained by a restricted best response moving a single packet away from $f$.
\begin{algorithm}[]
	\KwIn{equilibrium $x_k$ for $\G_{k}((d_i)_{i \in N})$, player $j$}
	\KwOut{equilibrium $x'_k$ for $\G_{k}((d'_i)_{i \in N})$, where $d'_j\leftarrow d_j+k; d'_i\leftarrow d_i$ for all $i \in N\setminus\{j\}$}
	$x \leftarrow x_k$; $d'_j\leftarrow d_j+k; \S'_j \leftarrow \S_j(d'_j,k)$; $d'_i\leftarrow d_i$ for all $i \in N\setminus\{j\}$\\
	Choose $f \in \arg\min_{e\in E}\{ \mu^{+k}_{j,e}(x)\}$;\\
	$x_{j,f} \leftarrow  x_{j,f} + k$;\\
	\While{$\exists i \in N$ who can improve in $\G_k$ }{
		Compute a restricted best response $y_i \in \S'_i$;\\
		$x_i \leftarrow y_i$;\\
	}
	$x'_k \leftarrow x$;\\
	\Return{$x'_k$}
	\caption{Subroutine $\ADD(x,j,\G_k((d_i)_{i \in N}))$}
	\label{addfunction}
\end{algorithm}  

\subsubsection{\RESTORE}
The second subroutine, $\RESTORE$, takes as input an equilibrium $x_{k}$ for packet size $k$ and game $\G_k((d_i)_{i \in N})$, and constructs an equilibrium for game $\G_{k/2}((d_i)_{i \in N})$ with packet size $k/2$. 
The algorithm relies on the sensitivity result of Theorem~\ref{thm:localload2}
bounding the load difference of every player under an equilibrium for the games $\G_k((d_i)_{i \in N})$
and $\G_{k/2}((d_i)_{i \in N})$, respectively. 
In particular, Theorem~\ref{thm:localload2} implies that for any two equilibria $x_k, x_{k/2}$
of the respective games, we have
$|(x_{k})_{i,e}-(x_{k/2})_{i,e}|\leq 2m^2nk$ for all $i\in N, e\in E$. Hence, we know already that
\emph{any} equilibrium $x_{k/2}$ satisfies $(x_{k/2})_{i,e}\geq y_{i,e}$ for all $i\in N, e\in E$, where
\begin{equation}\label{eq:set_y} y_{i,e}:=\max\{(x_{k})_{i,e}- 2m^2nk, 0\} \text{ for all $i\in N, e\in E$}.
\end{equation}
It follows that we can safely fix  $y_{i,e}, i\in N, e\in E$ according to~\eqref{eq:set_y}. The idea is to construct a  new game
with reduced demand of at most $4 m^3n^2$  packets of size $k/2$ and strategies $z_{k/2}\in \S_i(\bar d_i,k/2)$, where the reduced
demands are defined as
\[ \bar{d}_i = \sum_{e \in E} ((x_k)_{i,e} - y_{i,e}) \text{ for all $i \in N$ and $e \in E$.}\]

The private cost functions for players $i\in N$ are defined as
\begin{equation}\label{eq:hatpi} \hat \pi_i(z_{k/2}):= \sum_{e\in E}a_{i,e}  ((z_{k/2})_e+(y_{k/2})_e) ((z_{k/2})_{i,e}+(y_{k/2})_{i,e})+b_{i,e}((z_{k/2})_{i,e}+(y_{k/2})_{i,e}), \end{equation}
where $y_{k/2}:=y$ appears as a parameter. The form of $ \hat \pi_i(z_{k/2})$ in \eqref{eq:hatpi} is  designed to replicate the original cost function $ \pi_i(x_{k/2})$ for $x_{k/2}:=z_{k/2}+y_{k/2}$.
By multiplying out terms in \eqref{eq:hatpi}, it follows that the expressions $\sum_{e\in E}a_{i,e}(((z_{k/2})_{-i,e}+(y_{k/2})_e)(y_{k/2})_{i,e})+b_{i,e}(y_{k/2})_{i,e}$ are independent of $(z_{k/2})_i$ and can thus be left out (without losing equilibria).
We obtain a strategically equivalent game by replacing $\hat \pi_i(z_{k/2})$
with
\begin{align*} \bar \pi_i(z_{k/2})& := \sum_{e\in E}a_{i,e}   ((z_{k/2}+y_{k/2})_e (z_{k/2})_{i,e}+(z_{k/2})_{i,e}(y_{k/2})_{i,e})+b_{i,e}(z_{k/2})_{i,e}\\
&=\sum_{e\in E}a_{i,e}   (z_{k/2})_e (z_{k/2})_{i,e}+(b_{i,e}+a_{i,e}((y_{k/2})_e+(y_{k/2})_{i,e})) (z_{k/2})_{i,e}.
\end{align*}
Defining $\bar b_{i,e}:=b_{i,e}+a_{i,e}((y_{k/2})_e+(y_{k/2})_{i,e}), i\in N, e\in E$,  we obtain a standard
integer-splittable singleton congestion game with packet size $k/2$ and affine player-specific cost functions
$\bar c_{i,e}(z_e):= a_{i,e}  z_e+\bar b_{i,e}, i\in N, e\in E$.
This new game is denoted by 
\[\bar \G_{k/2}((\bar d_i)_{i \in N}) := \left(N,E, (\bar d_i)_{i \in N},  (\bar \pi_{i,e})_{i \in N, e\in E} \right).\]
$\RESTORE$ uses  now \textsc{PAlg} on  $\bar \G_{k/2}((\bar d_i)_{i \in N})$ to compute an equilibrium for game $\bar \G_{k/2}$. The key property leading to the claimed polynomial running time is that $\bar \G_{k/2}((\bar d_i)_{i \in N})$
 only involves a polynomially bounded number of packets. The pseudo-code of subroutine $\RESTORE$ can be found in Algorithm~\ref{restorefunction}.
\begin{algorithm}[h!]
	\KwIn{equilibrium $x_{k}$ for $\G_{k}((d_i)_{i \in N})$}
	\KwOut{equilibrium $x_{k/2}$ for $\G_{k/2}((d_i)_{i \in N})$}
	Define  $y_{i,e}:=\max\{(x_{k})_{i,e}- 2m^2nk, 0\}$ for all $i \in N$ and $e \in E$;\\  
	Define $\bar{d}_i = \sum_{e \in E} (x_k)_{i,e} - y_{i,e}$  for all $i \in N$ and $e \in E$;\\
	$z_{k/2}:=\textsc{PAlg}(\bar \G_{k/2}((\bar d_i)_{i \in N}) )$;\\
	\Return{$x_{k/2}:=y+z_{k/2}$;}
	\caption{Subroutine $\RESTORE(x_{k},\G_{k}((d_i)_{i \in N}))$.}
	\label{restorefunction}
\end{algorithm}

\subsubsection{\PH}
Using the subroutines $\ADD$  and $\RESTORE$, the algorithm $\PH$  computes an equilibrium $x_{k_0}$ for the $k_0$-splittable game $\G_{k_0}((d_i)_{i \in N}))$. In this algorithm, we start with an equilibrium $x_k$ for $\G_{k}((d'_i)_{i \in N}))$, where $d'_i=0$ for all $i \in N$, $k = 2^{q_1}k_0$ and $q_1 = \arg\min_{q \in \N} \{ 2^qk_0 > \max_{i \in N} d_i  \}$. Note that this game has a trivial equilibrium, where $(x_k)_{i,e}=0$ for all $i\in N$ and $e \in E$.  Then, we repeat the following two steps: 
\begin{enumerate}
	\item  Given an equilibrium $x_k$ for $\G_k((d'_i)_{i \in N})$, we construct an equilibrium for $\G_{k/2}((d'_i)_{i \in N})$ using subroutine $\RESTORE$ and set $k$ to $k/2$.
	\item For each player $i \in N$ we check if $d_i-d'_i\geq k$. If so, we increase $d'_i$ by $k$ and recompute equilibrium $x_k$ using subroutine $\ADD$.
\end{enumerate}
After $q_1$ iterations $\PH$ returns an equilibrium $x_{k_0}$ for $\G_{k_0}((d_i)_{i \in N}))$. The pseudo-code of $\PH$ can be found in Algorithm~\ref{bigalgorithm}.

\begin{algorithm}[h!]
	\KwIn{Integral splittable congestion game  $\G_{k_0}=(N,E,(d_i)_{i \in N},(c_{i,e})_{i \in N, e \in E})$.}
	\KwOut{An equilibrium $x_{k_0}$ for $\G_{k_0}$.}
	Initialize  $q_1 = \arg\min_{q \in \N} \{ 2^qk_0 > \max_{i \in N} d_i  \}$; 
	$k 	\leftarrow 2^{q_1}k_0$;  
	$d'_i \leftarrow 0$; $x_k \leftarrow (0)_{e \in E, i \in N}$;\\	
	\For{$1, \dots , q_1-1$}{	
		$x_{k/2} \leftarrow \RESTORE(x_{k},\G_k((d'_i)_{i \in N})))$;\\
		$k \leftarrow k/2$; \\	
		\For{$i \in N$}{
			\If{$d_i - d'_i>k$}{
				$x_k \leftarrow \ADD(x_k,i,\G_k((d'_i)_{i \in N})))$;\\
				$d'_i \leftarrow  d'_i + k$;\\}
		}		
	}
	\Return{$x_k$;}
	\caption{Algorithm $\PH(\G_{k_0}((d_i)_{i \in N}))$}
	\label{bigalgorithm}
\end{algorithm}

\section{Correctness}
\label{sec:correct}
In this section, we prove that $\PH$ indeed returns an equilibrium for game $\G_{k_0}((d_i)_{i \in N})$. In order to do so, we first need to verify that the two subroutines $\ADD$ and $\RESTORE$ are correct.  
Harks, Klimm and Peis~\cite[Thm.~5.1]{harks2016resource} proved that subroutine $\ADD$ indeed returns an equilibrium strategy for the new game with increased demand. It is left to verify correctness of $\RESTORE$ and $\PH$.

\subsection{Correctness \RESTORE}

By Lemma~\ref{groenevelt}, a strategy profile $z_{k/2}$ for the game $\bar \G_{k/2}$ is an equilibrium if and only if
\[ \bar \mu^{-k/2}_{i,e}(z_{k/2})\leq \bar \mu^{+k/2}_{i,f}(z_{k/2}) \text{ for all }i\in N, e,f\in E \text{ with }
(z_{k/2})_{i,e}>0,\]
where
$\bar \mu^{-k/2}_{i,e}(z_{k/2})$ and  $\bar \mu^{+k/2}_{i,e}(z_{k/2})$ denote the respective marginal costs for game $\bar \G_{k/2}$. 
We need to show that the combined profile $x_{k/2}:=(z_{k/2}+y_{k/2})$ satisfies
the conditions of Lemma~\ref{groenevelt} for the original game $\G_{k/2}$.
For this, we first relate the different marginal costs with each other.
\begin{lemma}\label{lem:marginal-bar}
Let $x_{k/2}:=(z_{k/2}+y_{k/2})$. Then, the following holds true:
\begin{enumerate}
\item  $\bar \mu^{-k/2}_{i,e}(z_{k/2})=\mu^{-k/2}_{i,e}(x_{k/2})$ for all $e\in E, i\in N$ with  $ (z_{k/2})_{i,e}>0$,
\item $\bar \mu^{+k/2}_{i,e}(z_{k/2})=\mu^{+k/2}_{i,e}(x_{k/2})$ for all $e\in E, i\in N$.
\end{enumerate}
\end{lemma}
\begin{proof}
Both statements follow by simple calculations. For $ (z_{k/2})_{i,e}>0$, we trivially have 
$ (x_{k/2})_{i,e}>0$ and thus we get:
\begin{align*}
\bar \mu^{-k/2}_{i,e}(z_{k/2})&= k(a_{i,e}((z_{k/2})_{i,e} + (z_{k/2})_e-k/2)+ \bar b_{i,e})\\
&= k(a_{i,e}((z_{k/2})_{i,e} + (y_{k/2})_{i,e}+(z_{k/2})_e+(y_{k/2})_e-k/2)+ b_{i,e})\\
&= k(a_{i,e}((x_{k/2})_{i,e} + (x_{k/2})_e-k/2)+ b_{i,e})\\
&= \mu^{-k/2}_{i,e}(x_{k/2}).
\end{align*}
The calculation for $\bar \mu^{+k/2}_{i,e}(z_{k/2})=\mu^{+k/2}_{i,e}(x_{k/2})$ works exactly the same by replacing $-k/2$ with $+k/2$.\ifllncs \qed \fi
\end{proof}

With Lemma~\ref{lem:marginal-bar} it follows that for an equilibrium $z_{k/2}$
of game $\bar \G_{k/2}$, the strategy profile $x_{k/2}:=(z_{k/2}+y_{k/2})$ is an equilibrium for $\G_{k/2}$
provided $(y_{k/2})_{i,e}>0\Rightarrow (z_{k/2})_{i,e}>0$ holds true for all $i\in N, e\in E$. 
For proving this, we first show a further sensitivity result for the  profiles $x_k$ and $(z_{k/2}+y_{k/2})$.
\begin{lemma}\label{lem:load_partial}
Let $x_k$ be an equilibrium for $\G_{k}$ and let $z_{k/2}$ be an equilibrium for the corresponding game  $\bar \G_{k/2}$.
Then, $|(z_{k/2})_e+(y_{k/2})_e-(x_k)_e|< 2nmk$ for all $e\in E$.
\end{lemma}
\begin{proof}
Assume $(z_{k/2})_e+(y_{k/2})_e-(x_k)_e\geq 2nmk$.
Define $E^+=\{e\in E\vert (z_{k/2})_e+(y_{k/2})_e-(x_k)_e\geq 0\}$ and $E^-=\{e\in E\vert (z_{k/2})_e+(y_{k/2})_e-(x_k)_e< 0\}$.
We get 
\[ \sum_{i\in N}\sum_{e\in E^+}  (z_{k/2})_{i,e}+(y_{k/2})_{i,e}-(x_k)_{i,e}\geq2nmk.\]
Hence, there is $i\in N$ with 
\begin{equation}\label{eq:newz} \sum_{e\in E^+}  (z_{k/2})_{i,e}+(y_{k/2})_{i,e}-(x_k)_{i,e}\geq 2mk > \tfrac{3}{2}mk.\end{equation}
Load balance implies
\[ \sum_{e\in E^-}  (z_{k/2})_{i,e}+(y_{k/2})_{i,e}-(x_k)_{i,e}\leq -2mk < - \tfrac{3}{2}mk.\]
From~\eqref{eq:newz} we get that there is $e\in E^+$ with 
\[  (z_{k/2})_{i,e}+(y_{k/2})_{i,e}-(x_k)_{i,e}>\tfrac{3}{2}k.\]
With $(y_{k/2})_{i,e}\leq (x_k)_{i,e}$ (cf.~\eqref{eq:set_y}), we have $(z_{k/2})_{i,e}>0$.
Moreover, there is $f\in E^-$ with 
\[ (z_{k/2})_{i,f}+(y_{k/2})_{i,f}-(x_k)_{i,f}<-\tfrac{3}{2}k.\]
It follows that $(x_k)_{i,f}>0$.
Altogether we obtain for some $i\in N, e\in E^+, f\in E^-$:
\begin{align*}
(z_{k/2})_e+(y_{k/2})_e-(x_k)_{e}+ (z_{k/2})_{i,e}+(y_{k/2})_{i,e}-(x_k)_{i,e}&>\tfrac{3}{2}k\\
(z_{k/2})_f+(y_{k/2})_f-(x_k)_{f}+ (z_{k/2})_{i,f}+(y_{k/2})_{i,f}-(x_k)_{i,f}&<-\tfrac{3}{2}k.
\end{align*}
With the usual calculation of marginal costs (as in the proof of Theorem~\ref{thm:globalload2}) and using Lemma~\ref{lem:marginal-bar} we get:
\begin{equation*}
	\bar \mu^{-k/2}_{i,e}(z_{k/2})
	> \frac{1}{2}\mu^{+k}_{i,e}(x_{k})  
	\geq \frac{1}{2} \mu^{-k}_{i,f}(x_{k})  
	> \bar \mu^{+k/2}_{i,f}(z_{k/2}).  
		\end{equation*}
		This contradicts the fact that $z_{k/2}$ is an equilibrium for game $\bar \G_{k/2}$. \ifllncs \qed \fi
\end{proof}
Now we are ready to prove that $x_{k/2}:=(z_{k/2}+y_{k/2})$ is an equilibrium for  $\G_{k/2}$.
\begin{lemma}
If $z_{k/2}$ is an equilibrium for game  $\bar \G_{k/2}$ w.r.t. an equilibrium $x_k$ for game $\G_{k}$, then 
$x_{k/2}:=(z_{k/2}+y_{k/2})$ is an equilibrium for  $\G_{k/2}$.
\end{lemma}
\begin{proof}
With Lemma~\ref{lem:marginal-bar},
it suffices to show that for all $e\in E, i\in N$:
\[ (y_{k/2})_{i,e}>0\Rightarrow (z_{k/2})_{i,e}>0.\]
Assume by contradiction that there is $i\in N$, $e\in E$
with $(y_{k/2})_{i,e}>0,  (z_{k/2})_{i,e}=0$.
By~\eqref{eq:set_y}, this implies
\[ (x_{k})_{i,e}-((z_{k/2})_{i,e}+(y_{k/2})_{i,e})= 2 m^2nk.\]
It follows that  $(x_{k})_{i,e}>0$.
With Lemma~\ref{lem:load_partial}, we get
\[ (x_{k})_{e}-((z_{k/2})_{e}+(y_{k/2})_{e}) \geq - 2mnk.\]
Hence, 
\[ (x_{k})_{e}-((z_{k/2})_{e}+(y_{k/2})_{e})+ (x_{k})_{i,e}-((z_{k/2})_{i,e}+(y_{k/2})_{i,e})\geq 2m(m-1)nk>2k.\]
Balance constraints imply
\[ \sum_{f\neq e} (x_{k})_{f}-((z_{k/2})_{f}+(y_{k/2})_{f})+ (x_{k})_{i,f}-((z_{k/2})_{i,f}+(y_{k/2})_{i,f})\leq - 2m(m-1)nk.\]
Thus, there is $f\neq e$ with 
\[ (x_{k})_{f}-((z_{k/2})_{f}+(y_{k/2})_{f})+ (x_{k})_{i,f}-((z_{k/2})_{i,f}+(y_{k/2})_{i,f})\leq - 2mnk<-2k.\]
Using $(x_{k})_{i,f}\geq (y_{k/2})_{i,f}$, the case $(z_{k/2})_{i,f}=0$ implies that the sensitivity bound of Lemma~\ref{lem:load_partial} applied on $f$
is violated, hence $(z_{k/2})_{i,f}>0$ must hold.
Applying the  same argumentation regarding the marginal cost (as in the proof of Theorem~\ref{thm:globalload2}) we get  that $z_{k/2}$ is not an equilibrium for game $\bar \G_{k/2}$
leading to a contradiction.
\ifllncs \qed \fi
\end{proof}

\subsection{Correctness \PH}
 It is left to prove that $\PH$ returns an equilibrium for  $\G_{k_0}((d_i)_{i \in N})$. 
\begin{theorem}
	Given a $k_0$-integral splittable singleton game with affine player-specific cost functions $\G_{k_0}:= \left(N,E,(d_i)_{i \in N}, (c_{i,e})_{i \in N,e \in E} \right)$, $\PH$ returns an equilibrium for $\G_{k_0}$. 
\end{theorem}
\begin{proof}
	We initialize  $x_{i,e}=0$ for all $i\in N$ and $e\in E$, which is an equilibrium for the game $\G_{2^{q_1}k_0}((0)_{i \in N})$. Assume that in iteration $q$ we enter the for-loop in $\PH$ with an equilibrium $x$ for game $\G_{2^{q_1-q+1}k_0}$ with demands  $d'_i = d_i - (d_i \mod 2^{q_1-q+1}k_0)$, where we use the notation $a \mod b := \alpha b $, where $\alpha=\arg\max\{z\in \Z_{\geq 0}\vert z b\leq a\}$. First, $\RESTORE$ computes an equilibrium for demands demands $d'_i $ and packet size $2^{q_1-q}k_0$. In lines 5-10 of $\PH$ we then check for each player $i\in N$, if her unscheduled load satisfies $d_i - d'_i \geq 2^{q_1-q}k_0$. If so, we schedule an extra packet for player $i$ using subroutine  $\ADD$. Note that at most one packet can be added per player and iteration.  Thus, after the $q$'th iteration in the for-loop, we obtain an equilibrium for demands $d'_i = d_i - (d_i \mod 2^{q_1 - q}k_0)$ and packet size $2^{q_1-q}k_0$. Hence, after the $q_1$'th iteration, we obtain an equilibrium for the desired packet size $2^0k_0 = k_0$ and demands $d'_i = d_i - (d_i \mod k_0) = d_i$, which is an equilibrium for game $\G_{k_0}((d_i)_{i \in N})$. \ifllncs \qed \fi
\end{proof}

\section{Running Time}
\label{sec:runningtime}
We prove that the running time of $\PH$ is polynomially bounded in $n$, $m$, $\log k$ and $\log \delta$, where $\delta$ is the upper bound on player-specific demands $d_i$. For this, we first need to analyze the running time of the two subroutines $\ADD$ and $\RESTORE$.

\subsection{Running Time \ADD}
 In~\cite[Theorem 5.2]{harks2016resource} Harks et al. proved that it takes time $nm(\delta/k)^2$ to execute $\ADD$. If their algorithm is applied to games with singleton strategy spaces and player-specific affine cost functions, we show next that the running time reduces to $O(nm^4)$. The main reason for this is that equilibria are not very sensitive under small changes in demands.  

\begin{lemma} \label{findingapath2} 
	Let $x_k$ be an equilibrium for game $\G_k((d_i)_{i \in N})$ and let $x_q$ be the strategy profile after the $q$'th iteration of the while-loop described in lines 4-7 of subroutine $\ADD$. Then $|(x_k)_{i,e} - (x_q)_{i,e}| < 2mk$ for all $i \in N$ and $e \in E$. 
\end{lemma}
\begin{proof}
	On the contrary, assume $q$ is the first iteration where $|(x_q)_{i,e}-(x_k)_{i,e}| = 2mk$ for some $i \in N$ and $e \in E$. There are two cases: either (I)$(x_q)_{i,e}-(x_k)_{i,e} = 2mk$ or (II)$(x_k)_{i,e}-(x_q)_{i,e} = 2mk$. We prove that the first case leads to a contradiction. For the second case a contradiction can be obtained in a similar manner.

	Harks,  Klimm and Peis~\cite{harks2016resource} proved that only the players
	using a resource whose load increased in the previous iteration
	may have an  improving move, and if so, a best response consists in moving one packet
	from this resource to another one. 
	This implies that $(x_k)_e \leq (x_q)_e \leq (x_k)_e+k$ for all $e \in E$. Thus, when assuming $(x_q)_{i,e} = (x_k)_{i,e} + 2mk$, we obtain:
\begin{equation} \label{cstarteq2} (x_q)_e + (x_q)_{i,e}  \geq (x_k)_e + (x_k)_{i,e}+ 2mk. \end{equation} 
	Remember that the total load distributed in $x_q$ by player $i$ exceeds the total load distributed in $x_k$ by at most $k$, and hence $\sum_{f \in E} (x_q)_{i,f} \leq k+ \sum_{f \in E} (x_k)_{i,f}$. We obtain:
	\[
	\textstyle \sum_{f \neq e} (x_q)_{i,f} \leq \sum_{f \neq e} (x_k)_{i,f} +(1-2m)k < \textstyle \sum_{f \neq e} (x_k)_{i,f} -2(m-1)k. 
	\]
	The pigeonhole principle implies that there exists  $ f\in E, f\neq e$ such that $(x_q)_{i,f} < (x_k)_{i,f} - 2k$ and thus  $(x_q)_{i,f} \leq (x_k)_{i,f} - 3k$. Combined with the fact that $(x_q)_f \leq (x_k)_f+k$, this implies: 
	\begin{equation} \label{ceq4} (x_q)_{i,f} + (x_q)_f \leq (x_k)_{i,f} + (x_k)_f - 2k. \end{equation}
	As $q$ is the first iteration in which $(x_q)_{i,e}-(x_k)_{i,e} = 2mk$, there is $e'\neq e$ 
	so that in iteration $q$, player $i$ moves a packet from  $e'$ to $e$, that is, 
	\[
	(x_q)_{-i}=(x_{q-1})_{-i} \text{ and }(x_q)_{i,g}=\begin{cases}(x_{q-1})_{i,g}+k, \text{ if } g=e,\\
	(x_{q-1})_{i,g}-k, \text{ if } g=e',\\
	(x_{q-1})_{i,g}, \text{ else}.\end{cases}
\]
	

	Using inequalities \eqref{cstarteq2}, \eqref{ceq4}, $m>1$ and the fact that $x_k$ is an equilibrium for packet size $k$, we obtain:
	\[\mu^{-k}_{i,e}(x_q) 
	> \mu^{+k}_{i,e}(x_k) 
	\geq \mu^{-k}_{i,f}(x_k)  
	\geq \mu^{+k}_{i,f}(x_q). \] 
	This, combined with the fact that $(x_q)_{i,e}>(x_k)_{i,e}\geq 0$ and that $(x_k)_{i,f} \geq (x_q)_{i,e} + 3k > 0$, implies player $i$ can decrease her cost by moving a packet from $e$ to $f$. This contradicts the fact that in strategy profile $x_{q-1}$ moving a packet to $e$ is a restricted best response for player $i$.	 \ifllncs \qed \fi
\end{proof}

\begin{lemma} \label{runningtimeadd}
	Algorithm $\ADD$ has running time $O(nm^4)$.
\end{lemma}
\begin{proof}
	Let $x_q$ be the strategy profile after line 5 of the algorithm has been executed for the $q$'th time, where we use the convention that $x_0$ denotes the preliminary strategy profile when entering the while-loop. Note that there is a unique resource $e_0$ such that $(x_0)_{e_0} = x_{e_0} + k$ and $(x_0)_e = x_e$ for all $e \in E \setminus{ \{e_0 \} }$. Furthermore, because we choose in Line 5 a restricted best response, a simple inductive argument shows that after each iteration $q$ of the while-loop, there is a unique resource $e_q \in E$ such that $(x_0)_{e_q} = x_{e_q} + k$ and $(x_0)_e = x_e$ for all $e \in E \setminus{ \{e_q \} }$.
	
We give each packet of size $k$ of the current demand of each player~$i \in N$ an identity
denoted by $i_j, j=1,\dots, r_i$ for some $r_i\in \N$.
We assume that players move packets according to a \emph{Last In First Out (LIFO)} principle. Thus, whenever player $i$ removes packet $i_j$ from $e_q$, she moves the packet that was placed on this resource last. We keep track of the marginal cost of a packet $i_j$ at the moment it is moved. Assume that packet $i_j$ is moved in $p$ iterations $q_1, \dots, q_p$. Then:	
	$$ \mu^{-k}_{i,e_{q_1}}(x_{q_1}) 
	 > \mu^{+k}_{i,{e_{q_1 + 1}}}(x_{q_1})
	= \mu^{-k}_{i,{e_{q_1 + 1}}}(x_{q_1 + 1})
	= \mu^{-k}_{i,e_{q_2}}(x_{q_{2}}).
	$$
	Here, the first equality is true as moving packet $i_j$ is an improving move for player $i$, 
	the second by construction of $x_{q_1 + 1}$ and the third as $e_{q_2}={e_{q_1 + 1}}$ and by LIFO principle $(x_{q_{2}})_{i,e_{q_2}}=(x_{q_1 + 1})_{i,e_{q_2}}$.
Applying this argument inductively, we obtain: $\mu^{-k}_{i,e_{q_1}}(x_{q_1})  > \mu^{-k}_{i,e_{q_2}}(x_{q_2})  > \dots > \mu^{-k}_{i,e_{q_p}}(x_{q_p}) .$ Note that in the iterations $q_1, \dots q_p$, the marginal cost value $\mu^{-k}_{i,e_{q_\ell}}(x_{q_\ell})$ does not depend on the aggregated load $(x_{q_\ell})_{e_{q_\ell}}$, as $(x_{q_\ell})_{e_{q_\ell}} = (x_q)_{e_{q_\ell}} + k$ for each $\ell \in \{1, \dots, p\}$. Instead it only depends on the player-specific load $(x_{q_\ell})_{i,e_{q_\ell}}$. Lemma~\ref{findingapath2} implies that each player $i \in N$ will move at most $2m$ packets from or to each resource and hence there will occur at most $4m$ different values of $(x_{q_\ell})_{i,e_{q_\ell}}$. Thus, each packet visits each resource at most $4m$ times. As each player $i$ moves at most $2m^2$ packets, and each packet visits each resource ($m$ resources) at most $4m$ times, the running time of $\ADD$ is bounded by $O(nm^4)$.  \ifllncs \qed \fi
\end{proof}

\subsection{Running Time \RESTORE}
\begin{lemma}\label{polysequence}
	$\RESTORE$ has running time $O(n^5m^{10})$.
\end{lemma}
\begin{proof}
Applying the algorithm of Harks et al.~\cite[Theorem 5.2]{harks2016resource} on the game $\bar \G_{k/2}$ yields a running time of
$n^2m(\frac{\delta }{k/2})^3$, where $n$ is the number of players, $m$ the number of resources, and
$\delta:=\max_{i\in N}\bar d_i$. By construction of the game $\bar \G_{k/2}$, the demands satisfy $d_i\leq 2m^3nk$ for all $i\in N$. Thus, we get $\delta \leq 2m^3nk$
yielding the desired result.\ifllncs \qed \fi
	\end{proof}

\subsection{Running Time \PH}
Finally, we prove the following theorem.
\begin{theorem}\label{thm:finaltheorem}
	$\PH$ runs in time $O(n^5m^{10} \log(\delta/k_0))$.
\end{theorem}
\begin{proof}
	Note that we picked $q_1 \in \N$ to be the smallest number such that $2^{q_1}k_0 >d_i$ for all player-specific demands $d_i$. This implies that $q_1$ is  in $O(\log (\delta/k_0))$, where $\delta:=\max_{i\in N}d_i$. Thus, the number of executions of  lines 3-9 is in $O(\log (\delta/k_0))$. In line 3, we call $\RESTORE$, which runs in $O(n^5m^{10})$. In lines $5-9$ we execute $\ADD$ (which runs in  $O(nm^4)$) at most $n$ times. Thus, the computation time of lines $4-9$ is $O(n^2m^4)$. This implies that the running of $\RESTORE$ dominates and it takes time $O(n^5m^{10})$ to go through a complete iteration in the for loop. Thus, $\PH$ runs in time $O(n^5m^{10} \log(\delta/k_0))$. \ifllncs \qed \fi
\end{proof}
Recall that for computing an atomic splittable equilibrium, we first compute the $k_0$ splittable equilibrium using the algorithm above. Second, we compute the exact equilibrium in time $O((nm)^3)$.
\begin{theorem} \label{thm:finalrunningtime}
	Given game $\G$, we can compute an atomic splittable equilibrium for $\G$ in  running time:
$
O\left((nm)^3 + n^5m^{10} \log(\delta/k_0) \right).
$
\end{theorem} 

\section{Multimarket Cournot Oligopoly}
\label{sec:cournot}
In this section, we derive a strong connection between atomic splittable singleton congestion games with affine cost functions and multimarket Cournot oligopolies with affine price functions and quadratic costs. Such a game is compactly represented by the tuple
\[ \mathcal{M}=(N,E,(E_i)_{i \in N}, (p_{i,e})_{i \in N,e \in E_i}, (C_i)_{i \in N}),\]
where $N$ is a set of $n$ firms and $E$ a set of $m$ markets. Each firm $i$ only has access to a subset $E_i \subseteq E$ of the markets. Each market $e$ is endowed with  firm-specific, non-increasing, affine price functions $p_{i,e}(t)=s_{i,e}-r_{i,e}t, i\in N$. In a strategy profile, a firm chooses a non-negative production quantity $x_{i,e} \in \R_{\geq 0}$ for each market $e \in E_i$. We denote a strategy profile for a firm by $x_i=(x_{i,e})_{e \in E_i}$, and a joint strategy profile by $x=(x_i)_{i \in N}$. The production costs of a firm are of the form $C_i(t)=c_it^2$ for some $c_i \geq 0$. The goal of each firm $i \in N$ is to maximize its utility, which is given by $u_i(x)=\sum_{e \in E_i} p_{i,e}(x_e)x_{i,e} -C_i\Big(\sum_{e \in E_i}x_{i,e}\Big)$, where $x_e:=\sum_{i \in N}x_{i,e}$. In the rest of this section we prove that several results that hold for atomic splittable equilibria and $k$-splittable equilibria carry over to multimarket oligopolies.

A strategic game $\G=(N,(X)_{i \in N},(u_i)_i \in N)$ is defined by a set of players $N$, a set of feasible strategies $X_i$ for each player $i \in N$ and a pay-off function $u_i(x)$ for each $i \in N$, where $x \in \bigtimes_{i \in N} X_i$.

\begin{definition}
	Let $\G=(N,(X_i)_{i \in N},(u_i)_i \in N)$, $\H=(N,(Y_i)_{i \in N},(v_i)_i \in N)$ be two strategic games with identical player set $N$. Then, $\G$ and $\H$ are called \emph{isomorphic}, if for all $i \in N$ there exists a bijective function $\phi_i:X_i \rightarrow Y_i$ and $A_i \in \R$ such that:
	$u_i(x_1, \dots x_n) = \nu_i(\phi_1(x_1), \dots ,\phi_n(x_n)) +A_i.$
\end{definition}

Let $\G=(N,(X_i)_{i \in N},(u_i)_i \in N)$ and $\H=(N,(Y_i)_{i \in N},(v_i)_i \in N)$ be isomorphic games. Then, $(x_i)_{i \in N}$ is an equilibrium of game $\G$ if and only if $(\phi_i(x_i))_{i \in N}$ is an equilibrium of game $\H$. This implies that $(x_i)_{i \in N}$ is the unique equilibrium of game $\G$ if and only if $(\phi_i(x_i))_{i \in N}$ is the unique equilibrium of game $\H$.

We prove that for each multimarket oligopoly, there exists an isomorphic atomic splittable game. Moreover, we can construct the isomorphism in polynomial time.

\begin{theorem}\label{thrm:contcournot}
	Given a multimarket oligopoly $\M$, there exists an atomic splittable game $\G$ that is isomorphic to $\M$.
\end{theorem}
\begin{proof}
	Given multimarket oligopoly $\M$, we construct an atomic splittable singleton game $\G$. For every firm $i \in N$ we create a player $i$ and we define her demand $d_i$ as an upper bound on the maximal quantity that firm $i$ will produce, that is, $d_i:=\sum_{e \in E_i}\max\{t \mid p_{i,e}(t)=0\}.$
	Note that if we limit the strategy space for each player $i \in N$ in game $\M$ to strategies $x$ satisfying $\sum_{e \in E_i} x_{i,e}\leq d_i$, all equilibria are preserved.
	Then, for every player $i$ we introduce a special resource $e_i$, and define the set of allowable resources for this player as: $\tilde{E}_i=E_i \cup \{e_i\} \; \text{ with $e_i \neq e_j$ for $i \neq j$}. $
The cost functions of special resources $e_i$ are defined as $c_{i,e_i}(t):=c_i(t-2d_i)$ for all $i \in N$ and
	the cost functions of resources $e \in E_i$ as: $c_{i,e}(t):=-p_{i,e}(t)=r_{i,e}t-s_{i,e}$ for all $i \in N$.
	In order to guarantee that the affine cost functions are non-negative, one can add a sufficiently large positive constant $c_{\max}$ to every cost function on each resource. We define 
	\[c_{\max}=\max\left\{ \{s_{i,e} \mid \text{ for all } i \in N, e \in E_i\} \cup \{2c_id_i \mid \text{ for all }i \in N\} \right\}.\]
	Note that adding $c_{\max}$ to every cost function does not change the equilibrium, it only adds $d_ic_{\max}$ to the total cost of each player. The total cost of a strategy $x$ for player $i$ in game $\G$ is: $\pi_i(x')=\sum_{e \in \tilde{E}_i} c_{i,e}(x'_e)x'_{i,e},$ which is equal to
	\begin{equation} \label{eq:bijectioncournot}
		\pi_i(x')=\textstyle \sum_{e \in E_i} -p_{i,e}(x'_e)x'_{i,e} + x'_{i,e_i}c_i(x'_{i,e_i}-2d_i).
	\end{equation}
	As maximizing pay-off equals minimizing costs, the payoff function of player $i$ in $x'$ is defined by: $v_i(x')=-\pi_i(x')$. It is left to prove that game $\G$ is isomorphic to game $\M$. Let $x$ be a feasible strategy in $\M$.  For each player $i \in N$, we define bijection $\phi_i:E_i \rightarrow \tilde{E}$ as: 
	$\phi_i(x_{i,1}, \dots, x_{i,m}) = (x_{i,1}, \dots, x_{i,m},d_i - \sum_{e \in E_i}x_{i,e}) =: (x'_{i,1}, \dots, x'_{i,m}, x'_{i,m+1}).$
	As we limited the strategy space for each $i \in N$ in game $\M$ to strategies $x$ where $\sum_{e \in E_i} x_{i,e}\leq d_i$, $x':=\phi(x)$ is a feasible strategy in $\G$. For each feasible strategy $x$ for game $\M$, and for each $i \in N$, we have: 
	\begin{eqnarray*}
		u_i(x)
		&=&\textstyle \sum_{e \in E_i} p_{i,e}(x_e)x_{i,e} - C_i\Big(\sum_{e \in E_i} x_{i,e} \Big)\\
		&=&\textstyle \sum_{e \in E_i} p_{i,e}(x_e)x_{i,e} - c_i \Big(d_i - \sum_{e \in E_i} x_{i,e} \Big)\Big(-d_i- \sum_{e \in E_i} x_{i,e} \Big) - c_id_i^2\\
		&=&\textstyle \sum_{e \in E_i} p_{i,e}(x_e)x_{i,e} - c_i \Big(d_i - \sum_{e \in E_i} x_{i,e} \Big)\Big(d_i- \sum_{e \in E_i} x_{i,e} -2d_i \Big) - c_id_i^2\\		
		&=&v_i(\phi_1(x_1), \dots, \phi_1(x_n))-c_id_i^2.
	\end{eqnarray*}
Thus, games $\M$ and $\G$ are isomorphic. \ifllncs \qed \fi
\end{proof}

One of our main results is our polynomial time algorithm that finds the unique equilibrium for atomic splittable singleton congestion games within polynomial time. As for each multimarket oligopoly there exists an atomic splittable game isomorphic to it, we can to construct this unique equilibrium within polynomial time.
\begin{theorem}
	Given a multimarket oligopoly $\M$, an equilibrium can be computed within running time: $ O\left(n^{15}m^{10}\log\left(\delta/k_0\right)\right).$
\end{theorem}
\begin{proof}
	This theorem follows directly from the fact that we can construct an atomic splittable singleton game $\G$ isomorphic to $\M$ (using Theorem~\ref{thrm:contcournot}) and the fact that $x=(x_i)_{i \in N}$ is an equilibrium in $\G$ if and only if $x=(\phi_i(x_i))_{i \in N}$ is an equilibrium in $\M$. Note that if in $\M$, firms compete over $m$ markets, the isomorphic atomic splittable singleton game $\G$ has $m+n$ resources. For such a game, Theorem~\ref{thm:finalrunningtime} implies that an equilibrium can be found in $O\left(n^3(m+n)^3+n^5(m+n)^{10}\log\left(\delta/k_0\right)\right).$ \ifllncs \qed \fi
\end{proof}

In an \emph{integral multimarket oligopoly} players sell indivisible goods. Thus, players can only produce and sell integer quantities, i.e., $x_{i,e} \in \N_{\geq 0}$ for each $i \in N$ and $e \in E_i$. For these games, we can construct an isomorphic 1-splittable congestion game.

\begin{theorem}\label{thrm:discreteisomorphic}
	Given an integral multimarket oligopoly $\M$, we can construct a 1-splittable congestion game $\G$ isomorphic to $\M$ within running time $O(nm)$.
\end{theorem}
\begin{proof}
	We define $d_i:=\sum_{e \in E_i} \lfloor \max\{t \mid p_{i,e}(t)=0\} \rfloor.$
	Then, the theorem follows using  the same construction as in Theorem~\ref{thrm:contcournot}. \ifllncs \qed \fi
\end{proof}
\begin{theorem}
	Given an integral multimarket oligopoly $\M$, an integral equilibrium can be computed within $O\left(n^{15}m^{10}\log\left(\delta/k_0\right)\right).$
\end{theorem}
\begin{proof}
	Theorem~\ref{thrm:discreteisomorphic} implies that we can construct an atomic splittable singleton game $\G$ isomorphic to $\M$. Note that if in $\M$, $n$ firms compete over $m$ markets, the isomorphic atomic splittable singleton game has $m+n$ resources. For such a game, Theorem~\ref{thm:finaltheorem} implies the desired running time. \ifllncs \qed \fi
\end{proof}

Lastly, we extend a result by Todd~\cite{Todd16}, where the total and individual production in one market in an integer equilibrium and a real equilibrium are compared. 

\begin{theorem}
	Given a multimarket oligopoly $\M$, with real equilibrium $(x_i)_{i \in N}$. Then, for any integer equilibrium $(y_i)_{i \in N}$ it holds that  $|x_e - y_e| \leq n(m+n-1)$ and  $|x_{i,e} - y_{i,e}| \leq n(m+n-1)(m+n)$.
\end{theorem}
\begin{proof}
	Assume that in game $\M$, $n$ firms compete over $m$ markets. According to Theorem~\ref{thrm:contcournot}, we can construct an atomic splittable congestion game $\G$ on $m+n$ resources that is isomorphic to $\M$ using the bijection $\phi$.  Let $x=(x_i)_{i \in N}$ be an atomic splittable equilibrium of $\M$ and let $y=(y_i)_{i \in N}$ be a $1$-splittable equilibrium of $\M$. Then $x':=(\phi_i(x_i))_{i \in N}$ is an atomic splittable equilibrium of $\G$ and $y':=(\phi_i(y_i))_{i \in N}$ is a $1$-splittable equilibrium of $\G$. According to Theorem~\ref{thm:globalload2} and~\ref{thm:localload2} we know that for any real equilibrium $x'$ and 1-splittable equilibrium $y'$ it holds that $|x'_e-y'_e|<n(m+n-1)$ and $|x'_{i,e}-y'_{i,e}|<n(m+n-1)(m+n)$ for all $i \in N$ and $e \in E_i$. Then, using the bijection $\phi$ described in~\eqref{eq:bijectioncournot}, we get $|x_e-y_e|<n(m+n-1)$ and $|x_{i,e}-y_{i,e}|<n(m+n-1)(m+n)$. \ifllncs \qed \fi 
\end{proof}

Todd~\cite{Todd16} showed that the total production in a 1-splittable equilibrium is at most $n/2$ away from that in the real equilibrium, and the individual firm's choice can be more than $(n-1)/4$ away from her choice in the real equilibrium. Our bounds are larger than Todd's, yet, they hold for a more general model -- multiple markets
and firm-specific price functions. We pose as an open question, whether or not our bounds are tight or can be further improved.

\section*{Notes and Acknowledgements}
This work is supported by the COST Action CA16228 ``European Network for Game Theory''.
We are grateful for the constructive comments of two anonymous reviewers.
We also thank Alex Skopalik and his students for pointing out an error of the proof of Theorem 1
appearing in previous versions of the paper. This updated version contains a correction of the error
and lead to some changed sensitivity bounds in Theorems 1 and 2. Moreover, in light of this
error, we had to change the subroutine $\RESTORE$ which now is simpler
and comes with a faster running time. All main results of the paper
contained in previous versions remain qualitatively intact (with some changed bounds
and corresponding running times though).


\bibliographystyle{spmpsci}      
\bibliography{masterbib}   

\end{document}